\documentclass[a4paper,UKenglish,cleveref, autoref, thm-restate]{article}

\usepackage{arxiv}
\usepackage{amsmath, amssymb}
\usepackage{tikz}
\usepackage{tikzscale}
\usepackage{graphicx}
\usetikzlibrary{shapes, arrows, automata, fit}
\usetikzlibrary{positioning}
\usepackage{mathtools}
\usepackage{hyperref}  
\usepackage{cleveref}

\usepackage{subfig}
\usepackage{todonotes}
\usepackage{multirow}
\usepackage{tabularx}
\usepackage{xspace}
\newcommand{\tabfig}{.15}

\usepackage{makecell}
\usepackage[linesnumbered, ruled, vlined]{algorithm2e}

\title{Wooly Graphs : A Mathematical Framework for Knitting}

\author{ 
	\href{https://orcid.org/0009-0006-0600-4477}{\includegraphics[scale=0.06]{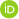}\hspace{1mm}Brian Bell} \\
	Department of Mathematics\\
	University of Arizona\\
	Tucson, AZ 85705 \\
	\texttt{bwbell@math.arizona.edu} \\
	\And
    \href{https://orcid.org/0000-0002-0621-5892}{\includegraphics[scale=0.06]{orcid.pdf}\hspace{1mm}Kathryn Gray} \\
	Department of Computer Science\\
	University of Arizona\\
	Tucson, AZ 85705 \\
	\url{http://rungray.github.io}\\
	\texttt{ryngray@arizona.edu} \\
	\And
	\href{https://orcid.org/0009-0003-7491-2811}{\includegraphics[scale=0.06]{orcid.pdf}\hspace{1mm}Diana Sieper} \\
	Institute of Mathematics\\
	Julius-Maximilians-Universität Würzburg\\
	Würzburg, Germany \\
	\texttt{marie.sieper@uni-wuerzburg.de} \\
	\And
	\href{https://orcid.org/0000-0002-0477-2724}{\includegraphics[scale=0.06]{orcid.pdf}\hspace{1mm}Stephen Kobourov} \\
	Technical University Munich\\
	Munich, Germany \\
	\url{https://www.professoren.tum.de/en/kobourov-stephen}\\
	\texttt{stephen.kobourov@tum.de} \\
	\And
	\href{https://orcid.org/0000-0002-9307-3254}{\includegraphics[scale=0.06]{orcid.pdf}\hspace{1mm}Falk Schreiber} \\
	Department of Computer and Information Science\\
	University of Konstanz\\
	Konstanz, Germany \\
	\texttt{falk.schreiber@uni-konstanz.de} \\
	\And
	\href{https://orcid.org/0000-0002-8345-5806}{\includegraphics[scale=0.06]{orcid.pdf}\hspace{1mm}Karsten Klein} \\
	Department of Computer and Information Science\\
	University of Konstanz\\
	Konstanz, Germany \\
	\texttt{karsten.klein@uni-konstanz.de} \\
	\And
	\href{https://orcid.org/0000-0003-1698-3868}{\includegraphics[scale=0.06]{orcid.pdf}\hspace{1mm}Seokhee Hong} \\
	School of Computer Science\\
	University of Sydney\\
	Sydney, Australia \\
	\texttt{seokhee.hong@sydney.edu.au} \\
}

\usepackage{graphicx} 
\usepackage{url, tabularx, environ}
\usepackage{amsmath, amssymb, amsthm}
\usepackage{wrapfig}
\usepackage{todonotes}
\usepackage[]{xcolor}
\usepackage[normalem]{ulem}

\date{March 2024}

\makeatletter
\newcommand{\problemtitle}[1]{\gdef\@problemtitle{#1}}
\newcommand{\probleminput}[1]{\gdef\@probleminput{#1}}
\newcommand{\problemquestion}[1]{\gdef\@problemquestion{#1}}

\NewEnviron{problem}{
  \problemtitle{}\probleminput{}\problemquestion{}
  \BODY
  \par\addvspace{.5\baselineskip}
  \noindent
  \begin{tabularx}{\textwidth}{@{\hspace{\parindent}} l X c}
    \multicolumn{2}{@{\hspace{\parindent}}l}{\@problemtitle} \\
    \textbf{Input:} & \@probleminput \\
    \textbf{Question:} & \@problemquestion
  \end{tabularx}
  \par\addvspace{.5\baselineskip}
}
\makeatother

\newtheorem{theorem}{Theorem}
\ifx\numberwithinsect\relax
  \numberwithin{theorem}{section}
  \fi

		\newtheorem{lemma}[theorem]{Lemma}

		\newtheorem{definition}[theorem]{Definition}

		\theoremstyle{definition}
		
		\theoremstyle{remark}
		\newtheorem{note}[theorem]{Note}
		\newtheorem*{note*}{Note}
		
		\newtheorem*{remark*}{Remark}
		\theoremstyle{claimstyle}
		
		\newtheorem*{claim*}{Claim}

\begin{document}

\maketitle

\begin{abstract}

This paper aims to develop a mathematical foundation to model knitting with graphs. We provide a precise definition for knit objects with a knot theoretic component and propose a simple undirected graph, a simple directed graph, and a directed multigraph model for any arbitrary knit object. Using these models, we propose natural categories related to the complexity of knitting structures. We use these categories to explore the hardness of determining whether a knit object of each class exists for a given graph. We show that while this problem is NP-hard in general, under specific cases, there are linear and polynomial time algorithms which take advantage of unique properties of common knitting techniques. This work aims to bridge the gap between textile arts and graph theory, offering a useful and rigorous framework for analyzing knitting objects using their corresponding graphs and for generating knitting objects from graphs.
\end{abstract}

\section{Introduction}

Knitting, an art form that weaves loops of yarn into intricate fabrics, has long fascinated both artists and mathematicians. The patterns created through knitting can be highly complex and often challenging to represent and analyze. Recent advances \cite{gray2024graphmodellayoutalgorithm, Counts_directed_2018} have explored the potential of graphs to model and understand these patterns, offering new perspectives and analytical tools. An example of the knitting, yarn path, and graph can be seen in figure~\ref{fig:wooly_graph}. This paper aims to bridge the gap in the literature by discussing methods from graph theory for identifying and categorizing graphs that represent knittable patterns. By defining and analyzing knitting graphs, we can determine whether a given graph corresponds to a knittable pattern. We further explore the properties of these graphs, including Hamiltonian and Eulerian paths, and propose a categorization of knitting techniques based on their graph complexity.

\begin{figure}[!htb]
    \centering
        \includegraphics[width=4.6cm]{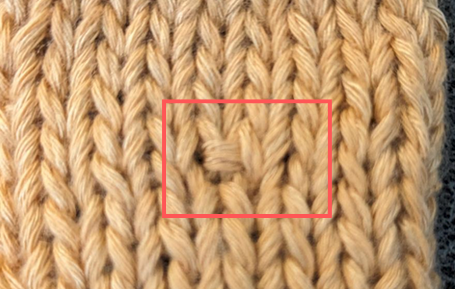}
\includegraphics[width=4.6cm]{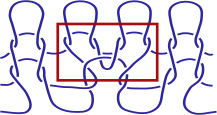}\resizebox{4.6cm}{!}{         
\begin{tikzpicture}

\node[circle, minimum size=19pt, fill=black!25, inner sep=0pt] (a11) at (0.0, 0.0) {$1$};
\node[circle, minimum size=19pt, fill=black!25, inner sep=0pt] (a21) at (1, 0.0) {$2$};
\node[circle, minimum size=19pt, fill=black!25, inner sep=0pt] (a31) at (2, 0.0) {$3$};

\node[circle, minimum size=19pt, fill=black!25, inner sep=0pt] (a12) at (-0.6, 1.0) {$7$};
\node[circle, minimum size=19pt, fill=red!50, inner sep=0pt] (a22) at (0.6, 1.0) {$6$};
\node[circle, minimum size=19pt, fill=red!50, inner sep=0pt] (a32) at (1.6, 1.0) {$5$};
\node[circle, minimum size=19pt, fill=black!25, inner sep=0pt] (a42) at (2.6, 1.0) {$4$};

\node[circle, minimum size=19pt, fill=black!25, inner sep=0pt] (a13) at (-0.6, 2.0) {$8$};
\node[circle, minimum size=19pt, fill=black!25, inner sep=0pt] (a23) at (0.6, 2.0) {$9$};
\node[circle, minimum size=19pt, fill=black!25, inner sep=0pt] (a33) at (1.6, 2.0) {$10$};
\node[circle, minimum size=19pt, fill=black!25, inner sep=0pt] (a43) at (2.6, 2.0) {$11$};

\draw[line width=0.02cm] (a11) -- (a21) -- (a31) -- (a42) -- (a43) -- (a33) -- (a32);
\draw[line width=0.02cm] (a33) -- (a23) -- (a22);
\draw[line width=0.02cm] (a23) -- (a13) -- (a12) -- (a22);
\draw[line width=0.02cm] (a12) -- (a11);
\draw[line width=0.03cm, color=red!75] (a42) -- (a32) -- (a22) -- (a21) -- (a32);    
\end{tikzpicture}}

    \caption{A common "Knit Front and Back" (kfb) stitch operation which increases the number of stitches in a row by 1 is shown in knitting (left), following the path of yarn (middle) and the resulting graph (right).}

    \label{fig:wooly_graph}
\end{figure}
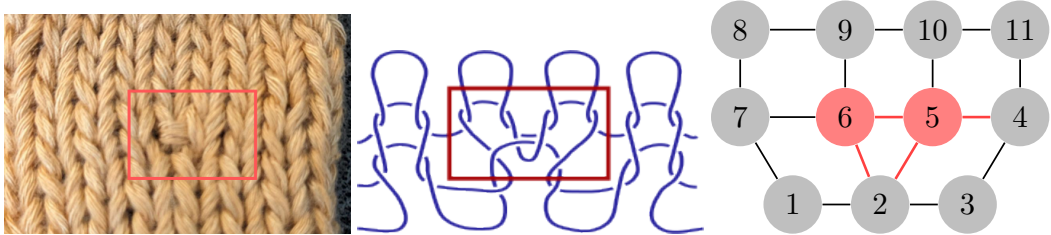

Through this exploration, we hope to provide a framework for understanding and modeling knitting patterns using graph theory, contributing to both the mathematical and artistic communities. In section~\ref{sec:related_work}, we discuss other crafts that have been modeled by graphs, as well as some previous work on knitting. In section~\ref{background}, we define knitting and knitting graphs as well as define some properties~\ref{properties} and define categories of knitting~\ref{categorizing}. In section~\ref{identifying}, we discuss determining whether a graph is knittable and the uniqueness of a given graph \ref{edge_cases}.  We give a discussion of the limitations of this approach in section~\ref{limitations} and conclude in~\ref{conclusion}.

\subsection{Related Work}\label{sec:related_work}

Blackwork Embroidery is an embroidery technique that comes from England.  It can be modeled as a graph \cite{Holden_embroidery_graph_2007}, and the paper goes over different properties of Blackwork Embroidery graphs, including defining how the graph must be Eulerian and has parity requirements on the edges.

Bobbin lace making has also been studied as it relates to graph theory \cite{Biedl_bobbin_2018}. The paper discusses how to identify a graph that could be a bobbin lace pattern and generate the bobbin lace from it.

Hitomezashi, an embroidery technique from Japan, can be encoded using two characters \cite{Seaton_hitomezashi_2022}.  This paper also discusses the duality property of Hitomezashi; a pattern on one side of the fabric corresponds to the reverse pattern on the other.  The authors note that this is also the case in knitting patterns. This is explored with a problem on reversibility in ``About Knitting" \cite{belcastro_knitting_2006}.  

Knitting itself has been studied by discussing knitting patterns in terms of recursive knitting patterns \cite{Bernasconi_recursive_knitting_2008} and how one could create a grammar around knitting. This led to a definition of the Kolmogorov complexity for different knitting patterns.  This model was based only on knit and purl stitches, looking at textures in knitting fabric.

 Counts models knitting with directed graphs~\cite{Counts_directed_2018} with a focus on visualization using a customized heuristic approach, as well as known graph embedding methods such as multidimensional scaling~\cite{coxon1972multidimensional, torgerson1952multidimensional} and the Kamada-Kawaii algorithm~\cite{kamada1989algorithm}. This approach has been further developed to better preserve edge lengths and match considerations of knitting geometry~\cite{gray2024graphmodellayoutalgorithm}

Another way to model knitting is by looking at the individual building blocks in terms of how the yarn is positioned \cite{markande2020}. This paper looks at the topological construction of knitting and breaks it down to the smallest repeatable section. 
While not about knitting, the book ``Figuring Fibers'' contains a chapter~\cite{calderhead_2018} that discusses a method to understand intermeshed crochet using graph theory. Here, we see the graphs used to understand how the two meshes fit together. A planar graph that can be crochet will give its dual as locations for the intermeshed portion. 

A generative method for creating knitting patterns, KnitGIST\cite{Hofmann_knitgist_2020}, takes different design considerations and stitch patterns to create hand-knitting patterns. This work uses a knitting graph but focuses on creating knitting patterns to fit a given stitch design and shape.

Many works look at machine knitting since it is usually easier to make things at scale with this technique. 
One paper that starts to bridge the gap \cite{Hofmann_knitpick_2019} begins with hand knitting and converts to machine knitting. This paper also creates methods to combine stitch techniques while keeping the general shape and stretch of the fabric.

The modeling of knitting using finite automata has also been studied~\cite{grishanov_1997}. This focuses on machine knitting and the state of the needles and yarn in each step. 

Machine knitting has also been modeled using meshes~\cite{narayanan2018}. Several rules for their graphs we allow in our framework, such as no crossings. Since we are focused on hand-knitting, we have different rules for what consists of feasible knitting. 

Another paper generates machine knitting patterns into arbitrary shapes using the mesh technique~\cite{Wu_stitchmeshes_2019}.  This paper~\cite{Yuksel_yarnlevel_2012}  creates a yarn-level visualization.  First, it generates a mesh that can correspond to different yarn positions. Then, it relaxes these positions so that the yarn is in a more realistic shape. This technique works for a variety of different stitches.

There is a paper looking at visually modeling knitting using edge length forces and normal forces to look at different knitting structures \cite{mckinlay_visualization_2023}. This paper begins with a stitch dictionary, defining how a mesh would be created from this. They also give examples of three-dimensional structures. 


\section{Background and Preliminaries}

In order to discuss knitting structures in detail, first we must carefully defining knit objects and second, we must relate these objects with graph models. Knitting objects have significant complexity depending on their context~\cite{kainen2023graph, Holden_embroidery_graph_2007, kaldor_2008}. In this study we will mostly consider the simplest models necessary to capture structure of knitted objects. From a minimal representation which treats each stitch as an infinitesimal vertex and its connections with other stitches as abstract edges in a graph, we will add complexity where necessary to examine knitting in other contexts, e.g. as knots. We will however be omitting the level of detail both in diameter of yarn and elastic properties necessary to analyze elastic properties of knit fabric.  

\subsection{Basic Definitions}\label{background}

We will begin by describing knit objects with mathematical precision:

\begin{definition}[$1$-Knit Object]
    A \emph{$1$-knit object} $O$ is an arrangement of interlocking loops in a single piece of yarn such that connecting the two ends of the piece of yarn creates the un-knot. Furthermore at least one loop must be secured only by another loop passing through it.
\end{definition}

This definition gives us the structure of knitting, loops or stitches, as seen in Fig:~\ref{fig:basic-knitting-single-stitch}. This figure shows a single row, in red, to show the path of the yarn.  In addition, we ensure that the knit object is not trivial by including that we must have one loop secured only by another loop. 

Now we will extend this definition by allowing more than one yarn. In knitting, techniques using more than one yarn are common, including fair isle, stripes, double knitting, intarsia, and other colorwork.

\begin{definition}[$k$-Knit Object]
    A \emph{$k$-knit object} $O$ is a collection of $k$ yarns arranged into loops which interlock and for which each independent yarn, joined end-to-end forms the un-knot. Furthermore at least one loop must be secured only by another loop passing through it. 
\end{definition}

Now we will explore a knit object that is finished. In knitting, a finished object refers to an object that is closed so that it does not unravel. We will show that this is a non-trivial.

\begin{definition}[$1$-Knit Finished Object]
    A $1$-Knit Finished Object is a 1-knit object which has been \emph{finished} such that when the loose ends are connected it forms a non-trivial knot.
\end{definition}

    Generally, this is made by passing the ends of yarn from a knit object a second time through their respective terminating stitches such that connecting these ends forms an irreducible knot with O(n) crossings where $n$ is the number of stitches. An example is shown in figure~\ref{fig:knot}. In this figure, we show a 5-coloring of the knitted object, showing that it is not the unknot.  In general, the shape of a knit stitch allows us to color both legs in one color and the piece behind with another. Then, we alternate these colors as we go up the knitting. More colors are necessary at the beginning and end of the knitting.

    \begin{figure}[ht]
    \centering
    \includegraphics[width=0.4\textwidth]{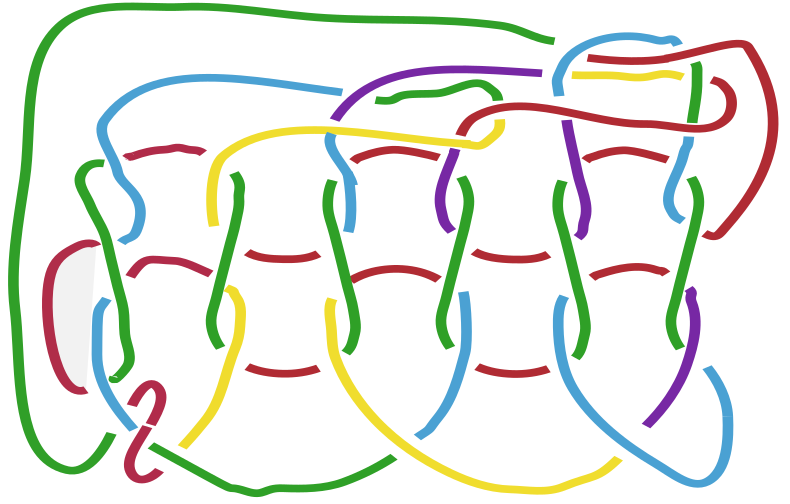}
    \caption{An example of a nontrivial knot formed by looping the loose ends of yarn again through their respective stitches and then connecting them. We color each stitch to show 5-colorability.}
    \label{fig:knot}
\end{figure}

\begin{definition}[$k$-Knit Finished Object]
    A $k$-Knit Finished Object is a knit object which has been \emph{finished} such that the loose ends of each piece of yarn can be connected to form (a) exactly $k$ complex Brunnian Links or (b) $n$ knots where $0 < n \le k$ and $k-n$ unknots such that none of the unknots are free.
\end{definition}

In case (a), we can construct a knit object such that the last unknot goes through each stitch on the last row. The more general case (b), at least one of the strands must not be the unknot. However many of these there are must hold the rest.


From here, we can start to define a knitting graph. We begin by defining what a vertex in relation to a knit stitch is.

\begin{definition}[Knitting Vertices]
    A \emph{knitting vertex} $v$ corresponds with one loop of yarn in a finished object which is either closed or held by another stitch passing through it. 
\end{definition}

A knitting vertex can be thought of as a stitch in knitting. A closed stitch occurs when one takes the tail end of the yarn and passes it directly through the stitch.  This keeps the loop from unraveling. Otherwise, the stitch will be held by another stitch.

Now we can define a knitting graph made up of these vertices.

\begin{figure}[ht]
\centering
\includegraphics[width=0.4\textwidth]{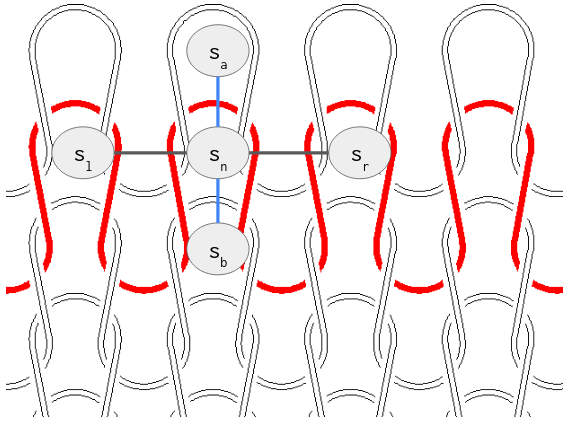}
\caption{An example of how a single node connects to the rest of the knitting graph. Here $S_n$ notes a single stitch with its connections.  $S_l$ and $S_r$ connect along the Hamiltonian path, which can also be seen in the red strand. $S_a$ and $S_b$ connect through a loop connection.}
\label{fig:basic-knitting-single-stitch}
\end{figure}

\begin{definition}[Knitting Graph]
    A \emph{knitting graph} $G = (V, E)$ for a knit object $O$ is a simple graph consisting of all knitting vertices $V$ and edges $E$ connecting each vertex to any other vertex that: (1) is sequential along the yarn with this stitch or (2) is passed through by this stitch. 
\end{definition}

We give an example of a single stitch with its graph connections in Fig:~\ref{fig:basic-knitting-single-stitch}. In that example, the two nodes to the right and left fall under (1), the stitches are sequential along the yarn.  We can see this by looking at the red thread. The nodes above and below fall under (2), each of these stitches is connected due to a stitch that is passed though. In the case of the node below, the current node is passed through and in the case of the node above, the above stitch is passed through the current node. In this way, we can create a knitting graph from a knit object by creating the nodes and defining edges based on the yarn connections. The sequential nature of stitches and the dependency of each stitch on another also induces an ordering which we will us later.

In addition, we can trace the path of the yarn through the knitting which shows one strand of yarn along the sequential connections and two strands of yarn for each of the vertical connections (one up, one down). We define a yarn graph to model these connections.

\begin{figure}[ht]
    \centering
    \includegraphics[width=0.4\textwidth]{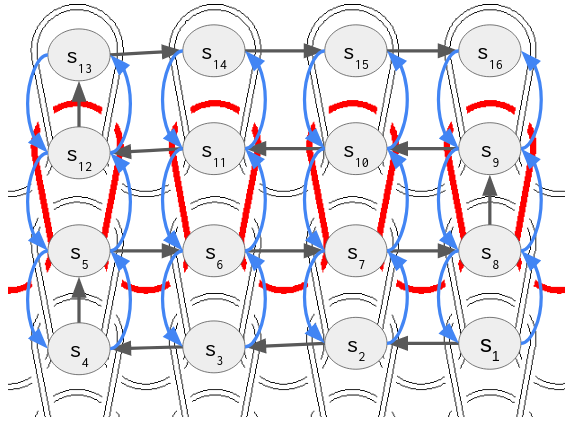}
    \caption{An example yarn graph for simple stockinette knitting. Here, the edges that are on the Hamiltonian are shown in black and the ones only on the Eulerian (yarn) path are shown in blue}
    \label{fig:eulerian_example}
\end{figure}

\begin{definition}[Yarn Graph]
The {yarn graph} $Y = (V, E)$ for a knit object $O$ is a  directed multi-graph consisting of all knitting vertices $V$ for the object and edges with edges $E$ for the path of the yarn as it passes through stitches.
\end{definition}

\begin{note} The knitting graph is an un-directed subset of the yarn graph.\end{note}

We can see an example of the yarn graph in Fig:~\ref{fig:eulerian_example}. In this example, we see that for each of the connections of type (2) in our knitting graph definition, we have an edge to and from that node. These edges represent the loop, yarn going to and from a given stitch.

Finally, we can compress the information from the yarn graph into a simple directed graph by color-coding edges corresponding with whether they are of type (1), i.e. single edges in the yarn graph, type (2), i.e. double edges in the yarn graph, or a mix of both (e.g. the turn stitches on the edges). We will use \textcolor{blue}{blue} for type (1), \textcolor{red}{red} for type (2), and \textcolor{purple}{purple} for mixed edges.  This is shown in figure~\ref{fig:color_graph}. Formally: 

\begin{definition}[Directed Knitting Graph]
The {directed knitting graph} $B = (V, E, C)$ for a knit object $O$ is a directed simple graph consisting of all knitting vertices $V$ for the object and edges with directed edges $E$ colors (red, blue, or purple) $C$. This graph reduces to the knitting graph for the same object by removing directions and colors. 
\end{definition}

\begin{figure}[ht]
    \centering
    \includegraphics[width=0.3\textwidth]{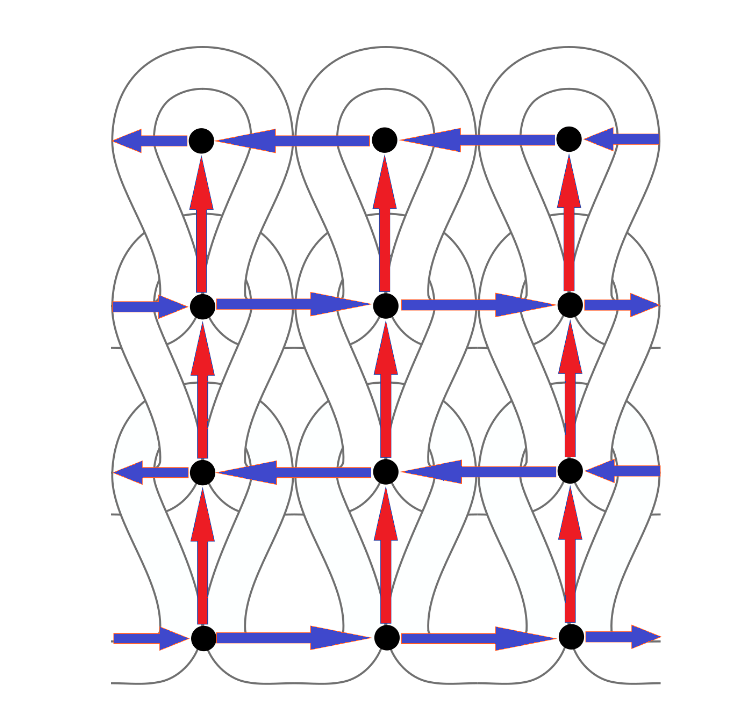}
    \caption{Example graph modeling from~\cite{Counts_directed_2018} (colors adjusted) The blue edges represent the yarn, the red edges represent interlocking (which is actually creating the `real' loops). Note purple edges are not visible in this section since they only appear on the exterior where the rows increment. }
    \label{fig:color_graph}
\end{figure}

\subsection{Properties of Knitting Graphs}\label{properties}

\begin{theorem}
    Every 1-knit object $O$ has a Hamiltonian path on its knitting graph $G$ and its yarn-graph $Y$. 
\end{theorem}

\begin{proof}Let $O$ be a 1-knit object with knitting graph $G$ and yarn graph $Y$. Since $O$ must unravel and become the un-knot when the ends are tied together, we will start at one end and follow the unraveling sequentially through each stitch. This is the Hamiltonian path on $G$. Following the same path, since the unraveling is sequential and exactly goes in the opposite of the yarn direction, this also induces a Hamiltonian path through $Y$ in reverse of the unraveling order. \end{proof}

\begin{note} The gap between the knitting needles passes through each stitch exactly once and thus follows the hamiltonian path induced by knitting an object. \end{note}

\begin{lemma} Every 1-knit object $O$ has an eulerian path on its yarn grpah $Y$
\end{lemma}

\begin{proof}

The the yarn induces a path through the yarn graph of the knit object, and since there is exactly one piece of yarn and every section of yarn has exactly one edge in the yarn graph, this path is Eulerian. \end{proof}

\begin{definition}
    A simple graph $H$ is $k$-knittable if there exists a $k$-knit object $O$ whose knitting graph $G = H$.
\end{definition}

\begin{theorem}
    Every graph with a Hamiltonian path is 1-knittable. 
\end{theorem}

\begin{proof} Let $G = (V,E)$ be a planar graph with a Hamiltonian path $P_G$... (proof by following the hamiltonian path and putting this stitch through any previous stitches to which it is connected. The hamiltonian ordering prevents collisions)\end{proof}

\begin{lemma}
    Every graph with a $k$-path cover k-knittable. 
\end{lemma}

\begin{proof} Suppose $G$ has a $k$-path cover, then construct knit-objects for the subgraph spanned by each path. Extend loops to connect the subgraphs. \end{proof} 

\begin{lemma}

\label{lem:cover}
    A graph is k-knittable if and only if it has a $k$-path cover. 
\end{lemma}
\begin{proof} suppose it doesn't have an $l$-path cover for $l \leq k$, then it cannot be knit with $k$ pieces of yarn without introducing a new edge connecting at least 2 of the paths in the cover, thus it cannot be knit with $k$ yarns. Suppose a graph is $k$-knittable. Then each piece of yarn corresponds with a path and since every stitch corresponds with a vertex, this induces a $k$-path cover. \end{proof}



\subsection{Categorizing Knitting}\label{categorizing}

Many knitting stitches and techniques exist, each bringing its own unique qualities to the craft. Basic stitches like the knit and purl form the foundation of most knitting projects, while more advanced techniques like cabling create textured patterns resembling twisted ropes. Lace involves creating delicate, openwork designs through a combination of increases and decreases. Colorwork techniques, such as Fair Isle and intarsia, enable the incorporation of multiple colors and yarns. Each of these methods expands the creative possibilities in knitting, allowing crafters to produce everything from simple to complex garments. Many of these techniques have bearing on our ability to define the graphs corresponding to these objects. In this section, we discuss how these techniques correspond to different classes of graphs. A brief summary can be found in table~\ref{tab:knitting_complexity}. 

We define the simplest knitting patterns and graphs with our complexity class 0. These correspond to planar graphs with restrictions enumerated in~\ref{fig:degree-configs}. In knitting, this includes techniques such as knit and purl stitches, increases, decreases, and simple lace. This does not include stitches which would knit many stitches into many stitches (e.g. knit 3 together through the front and back loop). 

\begin{figure}[!htb]
    \centering
    \includegraphics[width=.42\textwidth]{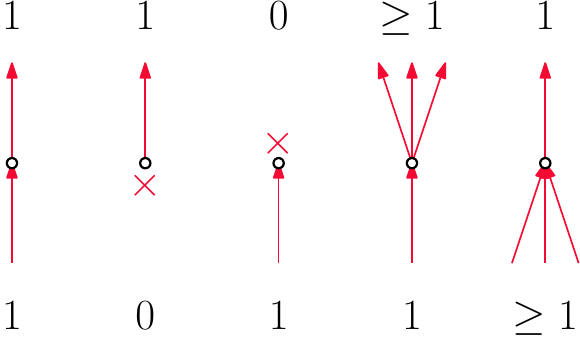}
    \caption{Allowed degree-distribution of the red edges.}
    \label{fig:degree-configs}
\end{figure}

Next, for class 1, we allow planar knitting without degree restrictions. This includes many-to-many stitches.

\begin{table*}[t]
    \centering
    \begin{tabular}{|c|c|c|c|}
         \hline Stitch  & In Fabric & Subgraph & Classification \\ \hline \hline
         \parbox[c]{\tabfig\textwidth}{Knit (k)} &  
         \parbox[c]{\tabfig\textwidth}{\includegraphics[width=\tabfig\textwidth]{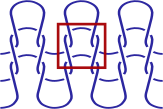}} & 
         \parbox[c]{\tabfig\textwidth}{
\resizebox{\tabfig\textwidth}{!}{         
\begin{tikzpicture}

\node[circle, minimum size=19pt, fill=black!25, inner sep=0pt] (a11) at (0.0, 0.0) {$1$};
\node[circle, minimum size=19pt, fill=black!25, inner sep=0pt] (a21) at (1.3, 0.0) {$2$};
\node[circle, minimum size=19pt, fill=black!25, inner sep=0pt] (a31) at (2.6, 0.0) {$3$};

\node[circle, minimum size=19pt, fill=black!25, inner sep=0pt] (a12) at (0.0, 1.0) {$6$};
\node[circle, minimum size=19pt, fill=red!50, inner sep=0pt] (a22) at (1.3, 1.0) {$5$};
\node[circle, minimum size=19pt, fill=black!25, inner sep=0pt] (a32) at (2.6, 1.0) {$4$};

\node[circle, minimum size=19pt, fill=black!25, inner sep=0pt] (a13) at (0.0, 2.0) {$7$};
\node[circle, minimum size=19pt, fill=black!25, inner sep=0pt] (a23) at (1.3, 2.0) {$8$};
\node[circle, minimum size=19pt, fill=black!25, inner sep=0pt] (a33) at (2.6, 2.0) {$9$};

\draw[line width=0.02cm] (a11) -- (a21) -- (a31) -- 
                         (a32) -- (a22) -- (a12) -- 
                         (a13) -- (a23) -- (a33) ;
\draw[line width=0.02cm] (a32) -- (a33) -- (a23) --
                         (a22) -- (a21) -- (a11) -- (a12);
\draw[line width=0.03cm, color=red!75] (a21) -- (a22) -- (a32);                         

\end{tikzpicture}}} & 
          \parbox[c]{\tabfig\textwidth}{Class 0} \\ \hline
         





          
          \parbox[c]{\tabfig\textwidth}{Yarn over (yo)} & 
          \parbox[c]{\tabfig\textwidth}{\includegraphics[width=\tabfig\textwidth]{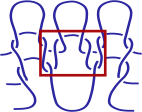}} & 
         \parbox[c]{\tabfig\textwidth}{
\resizebox{\tabfig\textwidth}{!}{         
\begin{tikzpicture}

\node[circle, minimum size=19pt, fill=black!25, inner sep=0pt] (a11) at (0.65, 0.0) {$1$};
\node[circle, minimum size=19pt, fill=black!25, inner sep=0pt] (a21) at (1.95, 0.0) {$2$};

\node[circle, minimum size=19pt, fill=black!25, inner sep=0pt] (a12) at (0.0, 1.0) {$5$};
\node[circle, minimum size=19pt, fill=red!50, inner sep=0pt] (a22) at (1.3, 1.0) {$4$};
\node[circle, minimum size=19pt, fill=black!25, inner sep=0pt] (a32) at (2.6, 1.0) {$3$};

\node[circle, minimum size=19pt, fill=black!25, inner sep=0pt] (a13) at (0.0, 2.0) {$6$};
\node[circle, minimum size=19pt, fill=black!25, inner sep=0pt] (a23) at (1.3, 2.0) {$7$};
\node[circle, minimum size=19pt, fill=black!25, inner sep=0pt] (a33) at (2.6, 2.0) {$8$};

\draw[line width=0.02cm] (a11) -- (a21) -- 
                         (a32) -- (a22) -- (a12) -- 
                         (a13) -- (a23) -- (a33) ;
\draw[line width=0.02cm] (a32) -- (a33) -- (a23) ;
\draw[line width=0.02cm] (a11) -- (a12);
\draw[line width=0.02cm] (a23) -- (a22);

\draw[line width=0.03cm, color=red!75] (a22) -- (a32);                         

\end{tikzpicture}}} & 
          \parbox[c]{\tabfig\textwidth}{Class 0} \\ \hline
          
          \parbox[c]{\tabfig\textwidth}{Knit Front and Back (kfb)} & 
          \parbox[c]{\tabfig\textwidth}{\includegraphics[width=\tabfig\textwidth]{stitches/kfb.png}} & 
          \parbox[c]{\tabfig\textwidth}{\resizebox{\tabfig\textwidth}{!}{         
\begin{tikzpicture}

\node[circle, minimum size=19pt, fill=black!25, inner sep=0pt] (a11) at (0.0, 0.0) {$1$};
\node[circle, minimum size=19pt, fill=black!25, inner sep=0pt] (a21) at (1, 0.0) {$2$};
\node[circle, minimum size=19pt, fill=black!25, inner sep=0pt] (a31) at (2, 0.0) {$3$};

\node[circle, minimum size=19pt, fill=black!25, inner sep=0pt] (a12) at (-0.6, 1.0) {$7$};
\node[circle, minimum size=19pt, fill=red!50, inner sep=0pt] (a22) at (0.6, 1.0) {$6$};
\node[circle, minimum size=19pt, fill=red!50, inner sep=0pt] (a32) at (1.6, 1.0) {$5$};
\node[circle, minimum size=19pt, fill=black!25, inner sep=0pt] (a42) at (2.6, 1.0) {$4$};

\node[circle, minimum size=19pt, fill=black!25, inner sep=0pt] (a13) at (-0.6, 2.0) {$8$};
\node[circle, minimum size=19pt, fill=black!25, inner sep=0pt] (a23) at (0.6, 2.0) {$9$};
\node[circle, minimum size=19pt, fill=black!25, inner sep=0pt] (a33) at (1.6, 2.0) {$10$};
\node[circle, minimum size=19pt, fill=black!25, inner sep=0pt] (a43) at (2.6, 2.0) {$11$};

\draw[line width=0.02cm] (a11) -- (a21) -- (a31) -- (a42) -- (a43) -- (a33) -- (a32);
\draw[line width=0.02cm] (a33) -- (a23) -- (a22);
\draw[line width=0.02cm] (a23) -- (a13) -- (a12) -- (a22);
\draw[line width=0.02cm] (a12) -- (a11);
\draw[line width=0.03cm, color=red!75] (a42) -- (a32) -- (a22) -- (a21) -- (a32);    
\end{tikzpicture}}} & 
          \parbox[c]{\tabfig\textwidth}{Class 0} \\ \hline
          
          \parbox[c]{\tabfig\textwidth}{Knit Two Together (k2tog)} & 
          \parbox[c]{\tabfig\textwidth}{\includegraphics[width=\tabfig\textwidth]{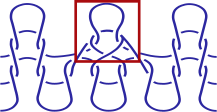}} & 
          \parbox[c]{\tabfig\textwidth}{\resizebox{\tabfig\textwidth}{!}{         
\begin{tikzpicture}


\node[circle, minimum size=19pt, fill=black!25, inner sep=0pt] (a12) at (0.0, 1.0) {$9$};
\node[circle, minimum size=19pt, fill=red!50, inner sep=0pt] (a22) at (1, 1.0) {$10$};
\node[circle, minimum size=19pt, fill=black!25, inner sep=0pt] (a32) at (2, 1.0) {$11$};

\node[circle, minimum size=19pt, fill=black!25, inner sep=0pt] (a11) at (-0.6, 0.0) {$8$};
\node[circle, minimum size=19pt, fill=black!25, inner sep=0pt] (a21) at (0.6, 0.0) {$7$};
\node[circle, minimum size=19pt, fill=black!25, inner sep=0pt] (a31) at (1.6, 0.0) {$6$};
\node[circle, minimum size=19pt, fill=black!25, inner sep=0pt] (a41) at (2.6, 0.0) {$5$};

\node[circle, minimum size=19pt, fill=black!25, inner sep=0pt] (a10) at (-0.6, -1.0) {$1$};
\node[circle, minimum size=19pt, fill=black!25, inner sep=0pt] (a20) at (0.6, -1.0) {$2$};
\node[circle, minimum size=19pt, fill=black!25, inner sep=0pt] (a30) at (1.6, -1.0) {$3$};
\node[circle, minimum size=19pt, fill=black!25, inner sep=0pt] (a40) at (2.6, -1.0) {$4$};

\draw[line width=0.02cm] (a11) -- (a21) -- (a31) -- (a41) -- (a32);
\draw[line width=0.02cm] (a12) -- (a22) -- (a32);
\draw[line width=0.02cm] (a12) -- (a11);
\draw[line width=0.03cm, color=red!75] (a12) -- (a22) -- (a21);  
\draw[line width=0.03cm, color=red!75] (a22) -- (a31);
\draw[line width=0.03cm] (a10) -- (a20) -- (a30) -- (a40);
\draw[line width=0.03cm] (a10) -- (a11);
\draw[line width=0.03cm] (a20) -- (a21);
\draw[line width=0.03cm] (a30) -- (a31);
\draw[line width=0.03cm] (a40) -- (a41);

\end{tikzpicture}}} & 
          \parbox[c]{\tabfig\textwidth}{Class 0} \\ \hline

          \parbox[c]{\tabfig\textwidth}{Cable One Behind (c1b)} & 
          \parbox[c]{\tabfig\textwidth}{\includegraphics[width=\tabfig\textwidth]{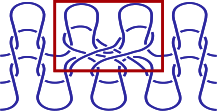}} & 
          \parbox[c]{\tabfig\textwidth}{\resizebox{\tabfig\textwidth}{!}{         
\begin{tikzpicture}

\node[circle, minimum size=19pt, fill=black!25, inner sep=0pt] (a11) at (0.0, 0.0) {$1$};
\node[circle, minimum size=19pt, fill=black!25, inner sep=0pt] (a21) at (1, 0.0) {$2$};
\node[circle, minimum size=19pt, fill=black!25, inner sep=0pt] (a31) at (2, 0.0) {$3$};
\node[circle, minimum size=19pt, fill=black!25, inner sep=0pt] (a41) at (3, 0.0) {$4$};

\node[circle, minimum size=19pt, fill=black!25, inner sep=0pt] (a12) at (0.0, 1.0) {$8$};
\node[circle, minimum size=19pt, fill=black!25, inner sep=0pt] (a22) at (1, 1.0) {$7$};
\node[circle, minimum size=19pt, fill=black!25, inner sep=0pt] (a32) at (2, 1.0) {$6$};
\node[circle, minimum size=19pt, fill=black!25, inner sep=0pt] (a42) at (3, 1.0) {$5$};

\node[circle, minimum size=19pt, fill=black!25, inner sep=0pt] (a13) at (0.0, 2.0) {$9$};
\node[circle, minimum size=19pt, fill=red!50, inner sep=0pt] (a23) at (1, 2.0) {$10$};
\node[circle, minimum size=19pt, fill=red!50, inner sep=0pt] (a33) at (2, 2.0) {$11$};
\node[circle, minimum size=19pt, fill=black!25, inner sep=0pt] (a43) at (3, 2.0) {$12$};


\draw[line width=0.02cm] (a12) -- (a11) -- (a21) -- (a31) -- (a41) -- (a42) -- (a32) -- (a31);
\draw[line width=0.02cm] (a42) -- (a43) -- (a33);
\draw[line width=0.02cm] (a13) --(a23);
\draw[line width=0.02cm] (a13) -- (a12) -- (a22);
\draw[line width=0.02cm] (a21) -- (a22) -- (a32);
\draw[line width=0.03cm, color=red!75] (a33) -- (a22);
\draw[line width=0.03cm, color=red!75] (a32) --(a23) -- (a33);

\end{tikzpicture}}} & 
          \parbox[c]{\tabfig\textwidth}{Class 2} \\ \hline

    \end{tabular}\medskip
    \caption{Overview of output from all algorithms on general knitting patterns.  The first column defines the knitting stitch we are looking at. The second column shows the actual knitted version.  The boxed area is the stitch and outside is plain knitting. The subgraph column shows the graph of the stitch with some plain stitches around it. We color the stitches of interest and edges created in red to distinguish them. When knitting flat, rows alternate directions, we begin each diagram going left to right. In the last column, we define the lowest category that the stitch could be a part of. Note that the stitches may be used in fabric that is in a higher category, e.g. a knit stitch used in the same fabric as a star stitch, in which case the knitting fabric is defined as the highest category of any stitch.}
    \label{tab:examples}
\end{table*}

\begin{table*}[t]
    \centering
    \begin{tabular}{|c|c|c|c|}
         \hline Stitch & Abbreviation & Strands of Yarn & Classification \\ \hline
         Knit & K & 1 & Class 0 \\ \hline
         Purl & P & 1 & Class 0 \\ \hline
         Yarn Over & YO & 1 & Class 0 \\ \hline
         Knit Front and Back & KFB & 1 & Class 0 \\ \hline
         Knit Two Together & K2TOG & 1 & Class 0 \\ \hline
         Slip Slip Knit & SSK & 1 & Class 0 \\ \hline
         Make One & M1 & 1 & Class 0 \\ \hline
         Knit Three Together & K3TOG & 1 & Class 0 \\ \hline
         Bobble & B(n) & 1 & Class 1\\ \hline
         Star Stitch & Star & 1 & Class 1 \\ \hline
         Knit One Below & K1B & 1 & Class 2 \\ \hline
         Knit One Slip One & K1SL1 & 1 & Class 2 \\ \hline
         Cable Two Behind & C2B & 1 & Class 2 \\ \hline
         Double Knit & DK & 2 & Class 3 \\ \hline
         Fair Isle & FI & 2 & Class 2 \\ \hline
         Intarsia & INT & 2 & Class 2 \\ \hline
         Brioche Knit & BRK & 1 & Class 2 \\ \hline
         Honeycomb Brioche Purl & HBP & 2 & Class 3 \\ \hline
    \end{tabular}
    \caption{List of knitting stitches and their abbreviations, required strands of yarn, and full names, strands of yarn needed, complexity classification.}
    \label{tab:stitch_abbreviations}
\end{table*}

The next complexity class, class 2, describes knit objects whose graph representation includes crossings and thus their edges must be labeled with an orientation (e.g. front/back). This class of knitting includes cables, color work, drop stitches, and many forms of brioche knitting. We can further refine this class by exploring the $k$-planarity of the underlying knitting graph. We define 1-planar knitting to be exactly all knit objects whose graphs are 1-planar. We can isolate two interesting subclasses based on the types of edges crossed: 1a-planar knitting requires no crossings in edges corresponding with stitches passing through other stitches (``vertical'' edges). 1b-planar knitting requires no crossings in edges corresponding with adjacency between stitches (``horizontal'' edges). Techniques corresponding to 1a-planar knitting include colorwork where yarn color changes every other stitch, drop stitches up to 1 row down, and simple brioche knitting. We note that brioche knitting -- which intentionally introduces extra crossings -- includes at least one technique which generate a maximally 1-planar graph when knitted in the round as shown in Fig.~\ref{fig:maximal}. This technique uses 4 yarns to produce extra crossings. 1b-planar knitting includes single cables (see c1b in Table~\ref{tab:examples} which cross at most 1 column. 

\begin{figure}
    \centering
\resizebox{0.35\textwidth}{!}{         
\begin{tikzpicture}

\node[circle, minimum size=19pt, fill=black!25, inner sep=0pt] (a11) at (0.0, 0.0) {$a$};
\node[circle, minimum size=19pt, fill=black!25, inner sep=0pt] (a21) at (1.3, 0.0) {$b$};
\node[circle, minimum size=19pt, fill=black!25, inner sep=0pt] (a31) at (2.6, 0.0) {$a$};
\node[circle, minimum size=19pt, fill=black!25, inner sep=0pt] (a41) at (3.9, 0.0) {$b$};
\node[circle, minimum size=19pt, fill=black!25, inner sep=0pt] (a51) at (5.2, 0.0) {$a$};
\node[circle, minimum size=19pt, fill=black!25, inner sep=0pt] (a61) at (6.5, 0.0) {$b$};

\node[circle, minimum size=19pt, fill=black!25, inner sep=0pt] (a12) at (0.0, 1.0) {$b$};
\node[circle, minimum size=19pt, fill=black!25, inner sep=0pt] (a22) at (1.3, 1.0) {$a$};
\node[circle, minimum size=19pt, fill=black!25, inner sep=0pt] (a32) at (2.6, 1.0) {$b$};
\node[circle, minimum size=19pt, fill=black!25, inner sep=0pt] (a42) at (3.9, 1.0) {$a$};
\node[circle, minimum size=19pt, fill=black!25, inner sep=0pt] (a52) at (5.2, 1.0) {$b$};
\node[circle, minimum size=19pt, fill=black!25, inner sep=0pt] (a62) at (6.5, 1.0) {$a$};

\node[circle, minimum size=19pt, fill=black!25, inner sep=0pt] (a13) at (0.0, 2.0) {$a$};
\node[circle, minimum size=19pt, fill=black!25, inner sep=0pt] (a23) at (1.3, 2.0) {$b$};
\node[circle, minimum size=19pt, fill=black!25, inner sep=0pt] (a33) at (2.6, 2.0) {$a$};
\node[circle, minimum size=19pt, fill=black!25, inner sep=0pt] (a43) at (3.9, 2.0) {$b$};
\node[circle, minimum size=19pt, fill=black!25, inner sep=0pt] (a53) at (5.2, 2.0) {$a$};
\node[circle, minimum size=19pt, fill=black!25, inner sep=0pt] (a63) at (6.5, 2.0) {$b$};

\draw[dotted, line width=0.02cm] (a13) .. controls (3, 2.91) .. (a63); 
\draw[dotted, line width=0.02cm] (a12) .. controls (3, 1.9) .. (a62); 
\draw[dotted, line width=0.02cm] (a11) .. controls (3, 0.9) .. (a61);
\draw[line width=0.02cm, color=blue!75] (a13) -- (a23) -- (a33) --
                         (a43) -- (a53) -- (a63) ;
\draw[line width=0.02cm, color=blue!75] (a11) -- (a21) -- (a31) --
                         (a41) -- (a51) -- (a61) ;
\draw[line width=0.02cm, color=red!75] (a12) -- (a22) -- (a32) --
                         (a42) -- (a52) -- (a62) ;                         
\draw[line width=0.02cm, color=red!75]  (a12) -- (a13);
\draw[line width=0.02cm, color=blue!75] (a22) -- (a23); 
\draw[line width=0.02cm, color=red!75]  (a32) -- (a33);
\draw[line width=0.02cm, color=blue!75] (a42) -- (a43);
\draw[line width=0.02cm, color=red!75]  (a52) -- (a53);
\draw[line width=0.02cm, color=blue!75] (a62) -- (a63);

\draw[line width=0.02cm, color=blue!75] (a11) -- (a12);
\draw[line width=0.02cm, color=red!75]  (a21) -- (a22); 
\draw[line width=0.02cm, color=blue!75] (a31) -- (a32);
\draw[line width=0.02cm, color=red!75]  (a41) -- (a42);
\draw[line width=0.02cm, color=blue!75] (a51) -- (a52);
\draw[line width=0.02cm, color=red!75]  (a61) -- (a62);

\draw[line width=0.02cm, color=blue!75] (a11) -- (a22) -- (a31) --
                         (a42) -- (a51) -- (a62) ;                       
\draw[line width=0.02cm, , color=red!75] (a12) -- (a21) -- (a32) --
                         (a41) -- (a52) -- (a61) ; 
\draw[line width=0.02cm, color=blue!75] (a13) -- (a22) -- (a33) --
                         (a42) -- (a53) -- (a62) ;                       
\draw[line width=0.02cm, color=red!75] (a12) -- (a23) -- (a32) --
                         (a43) -- (a52) -- (a63) ;

\end{tikzpicture}}
    \caption{Brioche knitting pattern made by holding 4 strands of yarn in 2 colors. Note that exterior stitches on the sides are connected i.e. ``in the round.''}
    \label{fig:maximal}
\end{figure}
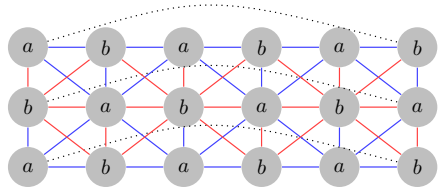

For $k > 2$, we note that the reach of each non-planar technique is increased. Each technique (colorwork, drop stitches, brioche techniques, and cables) has its reach increased by at least $k$ and as much as $2k$, meaning that drop stitches can drop at least $k$ rows down and possibly up to $2k$ rows, depending on the rest of the graph. There is a natural layout of each graph which corresponds to the plane in which it was knit. In this layout, the number of crossings introduced by each technique can be exactly determined by the parameters of the technique. For example in a cable 8 back (c8b), 8 stitches are removed from the needle and held to the back of the knitting while the subsequent 8 stitches are knit and then returned to the needle to be knit. This introduces 8 crossings in the natural knitting layout. However, a different layout may be found which introduces other crossings not naturally in the knitting in order to reduce the number of crossings corresponding with the cable in another layout. 

The complexity class 3 corresponds to knit objects whose vertices also require multiple orientations. This class includes double-knitting, where two layers of fabric (front and back) are knit simultaneously with two yarns allowing the yarns to be swapped arbitrarily between the front and back. This allows for very sophisticated color work and produces a thicker combined fabric. Other techniques,  including some forms of brioche, also fit this classification. Depending on the level of detail required, this class of knitting may require 3D layouts or impose stronger requirements on 2D layouts for visualization. 

We can get some insight into these classes by looking at a popular knitting and crochet website, Ravelry~\cite{Ravelry}. Ravelry has a search function where you can look through hundreds of thousands of patterns. Of these, about 240,000 are in class 0, 16,000 in class 1, 105,000 in class 2, and 6,000 in class 3. Some examples of stitches which delineate the differences between these classes are shown in~\ref{tab:stitch_abbreviations}. This quick look did not include every pattern on Ravelry since Ravelry's categories are optional for patterns and some categories on Ravelry could fall into multiple classes.


\begin{table}[]
    \centering
    \begin{tabular}{|c|c|c|}
        \hline Complexity & Knit Description & Graph Description  \\\hline
        0 & Basic Simple Knitting & planar graphs with restrictions from ~\ref{fig:degree-configs}\\ \hline
        1 &  Common simple knitting & planar graphs \\ \hline
        2 & $k$-planar knitting & \parbox{3.4cm}{Edges have Two Orientations (cables, colorwork, brioche)} \\ \hline
        3 & Double Knitting & \parbox{3.4cm}{Edges and Nodes have\\ multiple orientations}  \\ \hline
        
    \end{tabular}
    \caption{Description of categories of knitting in relation to their graph description.}
    \label{tab:knitting_complexity}
\end{table}

\section{Identifying Knittable Graphs}\label{identifying}



Now that we have labeled different types of knitting, we will focus on identifying which graphs are knittable in the various complexity classes that we have defined. We will start with the general graph case.

\begin{lemma}
\label{lem:identify-k-knittable}
Identifying general $k$-knittable graphs is NP-hard.
\end{lemma}

\begin{proof} Let $G$ be a graph. By Lemma~\ref{lem:cover}, $G$ is k-knittable iff $G$ has a $k$-path cover. Determining whether $G$ has a $k$-path cover is NP-hard ~\cite{karp2010reducibility}. \end{proof}

\subsection{Identifying $k$-knittable DAGs.}

If we have more information on the graph and know that it is a DAG, we can identify a knittable graph easier. We start with the definition of a $k$-knittable DAG.

 \begin{definition}
    A directed graph $H$ is $k$-knittable if there exists a coloring $C$ and a $k$-knit object $O$ whose directed knitting graph $G = H \bigoplus C$.
\end{definition}

\begin{theorem}
    1-knittable DAGs can be identified in linear time. 
\end{theorem}
We can find a Hamiltonian path, if it exists, in linear time.  By sorting the nodes topologically, we end up with an ordering for the Hamiltonian path.  As shown above, given a Hamiltonian path, we can construct a knitting pattern, by setting every edge not on the Hamiltonian to a loop connection. \qed

\begin{note}
    If there are restrictions on the edge connections, we can check this as we build the Hamiltonian by checking the in and out degrees.
\end{note}

Likewise, a $k$-knittable DAG follows similarly.

\begin{theorem}
    $k$-knittable DAGs can be identified in polynomial time.
\end{theorem} 

This is equivalent to finding a k-disjoint-path-cover, which has been studied in a variety of settings. It has been shown to be parameterized linear \cite{Caceres_mpcdag_2022}, although the general algorithm is polynomial through a reduction to minimum flow \cite{Ntafos_pathcover_1979}.

\begin{theorem}
    \label{thm:k-knittable-DAGs-no-many-many}
    It can be decided in polynomial time, whether a DAG is $k$-knittable with no violet edges, and of complexity class 0.
\end{theorem}

\begin{proof}
When deciding whether a DAG $G=(V,E)$ is $k$-knittable, the task is to color the edges blue or red in such a way the blue edges form at most~$k$ directed paths, and the red edges fulfill the degree restrictions for every vertex.
In order to test this we generate a flow network for~$G$ in which a valid flow of value $k$ will represent the~$k$ threads.
For every vertex~$v$, we can decide whether it can be the start of a thread ($S$), at the end of a thread ($T$), in the middle ($M$) or any combination of these simply by looking at the in- and outdegree of $v$.
We define $S\subseteq V$ to be the vertices $v\in V$, such that $\mathrm{indeg}(v)$ and $\mathrm{outdeg}(v)-1$ fulfill the red degree restrictions.
Next, we define $M \subseteq V$ ($T\subseteq V$) to be the vertices, such that $\mathrm{indeg}(v)-1$ and $\mathrm{outdeg}(v)-1$ ($\mathrm{indeg}(v)-1$ and $\mathrm{outdeg}(v)$) fulfill the red degree restriction.
Note that this $M$, $S$ and $T$ are generally partially overlapping.
Further, if there is any vertex that is in none of the sets $M$, $S$ and $T$, the solution is infeasible.
\Cref{tab:degrees-overview} sums up in which sets a vertex is, based on its degrees,
\begin{table}[htb]
    \centering
    \begin{tabular}{|c|c|c|c|c|}
    \hline
    Indeg  & 0 & 1 & 2 & $\geq 3$ \\
    \hline
    Outdeg & & & &  \\
    \hline
    0 & non-feasible & non-feasible & $S$  & non-feasible \\
    \hline
    1 & non-feasible & non-feasible & $M$, $T$ & $T$ \\
    \hline
    2 & $S$  & $S$, $M$  & $S$, $M$, $T$ & $M$ \\
    \hline
    $\geq 3$ & non-feasible & $S$  & $M$  & non-feasible \\
    \hline
    \end{tabular}
    \caption{Combinations of in- and outdegrees and their implications on feasibility for \Cref{thm:k-knittable-DAGs-no-many-many}.}
    \label{tab:degrees-overview}
\end{table}

We generate a flow graph~$G_F$ out of~$G$. First, we split every vertex $v$ into two vertices $v_\mathrm{in}$, $v_\mathrm{out}$. If $v\in M$, we assign all incoming edges of $v$ to $v_\mathrm{in}$, all outgoing edges to $v_\mathrm{out}$, and we add a new directed edge from $v_\mathrm{in}$ to $v_\mathrm{out}$, to which we assign a minimum and a maximum flow of one (enforcing that precisely one unit of flow and thus one thread will contain $v$).
Note that if $v\notin M$, then either $v\in S$ or $v\in T$.
If $v\in S$, we only assign the outgoing edges of $v$ to $v_\mathrm{out}$ (with maximum capacity 1), but we omit assigning the incoming edges of $v$.
Similarly, if $v \in T$, we only assign the incoming edges of $v$ to $v_\mathrm{in}$ (with maximum capacity 1), but we omit assigning the outgoing edges of $v$.
We add two new source-vertices~$s_\mathrm{in}$ and~$s_\mathrm{out}$, and connect~$s_\mathrm{in}$ to~$s_\mathrm{out}$ via a directed edge with capacity of precisely~$k$ (ensuring that any valid flow will have a value of precisely~$k$, ensuring that precisely~$k$ threads are used).
Further, we add edges from~$s_\mathrm{out}$ to every vertex~$v \in S$ with maximum capacity~1.
Similarly, we add two new sink-vertices~$t_\mathrm{in}$ and~$t_\mathrm{out}$, and connect~$t_\mathrm{in}$ to~$t_\mathrm{out}$ via a directed edge with capacity of precisely~$k$.
Further, we add edges from every vertex~$v \in T$ to~$t_\mathrm{in}$ with maximum capacity~1.
All edges that were originally present in~$G$ get assigned a maximum capacity of~1.

Given an s-t-flow $F$ for $G_F$, we can extract $k$ threads that cover $V$ by following the flows that leave $s_\mathrm{out}$ until they arrive at $s_\mathrm{in}$, taking for every unit of flow all vertices $v$ for which this unit uses the edge $(v_\mathrm{in}, v_\mathrm{out})$ in this order.

Then, since for every $v\in V$ the minimum and maximum capacity of $(v_\mathrm{in}, v_\mathrm{out})$ is precisely 1, every vertex is covered by precisely one thread.
Also, we ensured that only vertices that are in $S$ can be the start of a thread, and only vertices that are in $T$ can be the end of a thread.
Further, since for every vertex $v \notin M$ we either added only outgoing edges or only incoming edges apart from the edges from/to $s_\mathrm{out}, t_\mathrm{in})$, none of these vertices can be in the middle of a thread.
Thus, the degree constraints are fulfilled for every vertex.

By taking a thread cover of $G$ of size $k$, we can construct a valid flow for $G_F$. \end{proof}

\subsection{Edge Cases}\label{edge_cases}
\begin{theorem}
    Each isomorphically distinct Hamiltonian path gives a different knit object. 
\end{theorem}

    \begin{proof}Given two isomophically distinct Hamiltonian paths, with nodes $H1_1 \dots H1_n$ and $H2_1 \dots H2_n$, there must be edges $(H1_k, H1_j)$ that correspond to $(H2_m, H2_n)$.  Since $k \ne m$ and $j \ne n$, the looped stitches will be attached in different places.  This gives two different knitted object.\end{proof}

    In terms of a DAG, only one Hamiltonian path is possible, which reduces the total number of knit objects.  However, if we do not have a DAG, the number reduces to finding Hamiltonian paths and is NP-Hard.







\begin{theorem}
    The number of rows which can be knit [using basic knitting stitches] for a 1-knittable object which has a planar directed knitting graph can be computed in linear time. 
\end{theorem}

\begin{proof}
    Given a planar directed knitting graph corresponding with a 1-knittable object, we can follow a topological ordering (this will follow the Hamiltonian path of the needles through the object). Starting with our first source in the topological ordering, since the graph is planar, we can label all edges originating from this vertex which are not part of the hamiltonian path ``left". Now, we will proceed until a source edge appears on the opposite side which we will call ``right". Each time the side switches, we will add another row. Since the graph is planar and we have incremented rows each time, we have not violated the conditions for basic knitting. This can be seen if we look at book embeddings of planar, Hamiltonian graphs. It is well known that such graphs have a book embedding of 2 pages \cite{Bernhart_book_1979}, now whenever a node's outgoing edges switch sides, we define this as the next row in knitting.
\end{proof}

This does gives some ambiguous cases when nodes do not contain outgoing edges. Either a row switch or maintaining the same row would be acceptable, we would suggest keeping the row to minimize the number of row switches.




\section{Discussion and Limitations}\label{limitations}

We model knitting with a discrete structure, while knitting itself consists of a continuous yarn. This gives us a way to model the knitting, however, it may not model everything that we would want in the knitting. For example, the stretch of the knitting due to the yarn is not modeled in the graph form. A more sophisticated model will be necessary to study elasticity which is a direction for future work.

We create several classes of knitting, which we believe covers most knitting stitches. There are many different stitch varieties and while we believe all of them can fit into these categories there may be stitches that require a new category.

Our graph representation of knitting is referenced to the top of each knit loop which leaves an ambiguity about how the graph should be indexed especially with regards to the cast-on row and edges. 


While we discuss how knit and purl stitches are different and result in different orientations in the yarn graph, we do not explore how this affects the fabric. Looking at the reverse of the knit fabric and modeling the different sides with different edge orientations is an interesting future work.

Exploring how cables affect the space of the knitting is a future direction that can allow us to understand the topology of graphs as well.

Finally, double knitting and brioche are two techniques that have interesting structure and many knitting patterns and techniques associated with them. Creating graphs to model these specific stitches can give us a better understanding of knitting and possible yarn connections.

\section{Conclusions and Future Directions}\label{conclusion}

This paper looks at different models of knitting using graphs. We have defined several categories of knitting that correspond to different types of graph. These categories also follow general knitting wisdom of complexity. From this point, we determine several properties of graphs that are necessary for them to be knitting graphs. We discuss the uniqueness of knitting graphs and how many knitting patterns we could have from one graph.

While we look at each of these things for some of our classifications, there are many open questions in this direction. For example, double knitting has some very strict structural requirements. Defining and modeling these in terms of the necessary covers and node connections is an interesting problem. Likewise, brioche has complex yarn and knitting connections that would lead to some interesting algorithms. Finally, measurement of elasticity of knit fabric is a well motivated future work which we hope can be built on top of our framework with some additional modeling information.

\bibliography{bib}

\end{document}